\documentclass[twocolumn,amsthm]{autart}
\usepackage{graphicx}
\usepackage[utf8]{inputenc}
\usepackage{amsmath}
\usepackage{amsfonts}
\usepackage{mathtools}
\usepackage{amssymb}
\usepackage{comment}
\usepackage{color}
\usepackage{upgreek}
\usepackage{float}
\usepackage{xcolor}
\usepackage{soul}
\usepackage{hyperref}
\usepackage{tikz-cd}
\usepackage{balance}
\usepackage[inline]{enumitem}

\mathtoolsset{showonlyrefs}

\theoremstyle{plain}

\newtheorem{proposition}{Proposition}[section]
\newtheorem{corollary}{Corollary}[section]
\theoremstyle{definition}
\newtheorem{lemma}{Lemma}[section]
\newtheorem{assumption}{Assumption}

\newtheorem{remark}{Remark}

\def\R{\mathbb{R}}
\def\N{\mathbb{N}}

\usepackage[normalem]{ulem}

\def\JPC{{\it J. Process Control}}

\def\IJACSP{{\it Int. J. Adapt. Control Signal Process.}}
\def\IJRNLC{{\it Int. J. Robust Nonlinear Control}}
\def\TAC{{\it IEEE Trans. Autom. Control}}
\def\TIE{{\it IEEE Trans. Ind. Electron.}}

\def\EJC{{\it Eur. J. Control}}

\def\SCL{{\it Syst. Control Lett.}}
\def\AUT{{\it Automatica}}
\def\CST{{\it IEEE Trans. Control Syst. Technol.}}

\def\CSM{{\it IEEE Control Syst. Mag.}}

\def\TPE{{\it IEEE Trans. Power Electron.}}
\def\CSL{{\it IEEE Control Syst. Lett.}}

\def\lab{\label}
\def\liminf{\lim_{t \to \infty}}

\def\begequ{\begin{equation}}
\def\endequ{\end{equation}}

\def\begmat#1{\begin{bmatrix}#1\end{bmatrix}}
\def\begali#1{\begin{align}{#1}\end{align}}
\def\begalis#1{\begin{align*}{#1}\end{align*}}

\def\begcen{\begin{center}}
\def\endcen{\end{center}}

\def\liminf{\lim_{k \to \infty}}

\def\L2{{\cal L}_2}
\def\L2e{{\cal L}_{2e}}

\def\begmat#1{\begin{bmatrix}#1\end{bmatrix}}
\def\begali#1{\begin{align}{#1}\end{align}}
\def\begalis#1{\begin{align*}{#1}\end{align*}}

\def\begequarr{\begin{eqnarray}}
\def\endequarr{\end{eqnarray}}
\def\begequarrs{\begin{eqnarray*}}
\def\endequarrs{\end{eqnarray*}}
\def\begarr{\begin{array}}
\def\endarr{\end{array}}
\def\begequ{\begin{equation}}
\def\endequ{\end{equation}}
\def\lab{\label}
\def\begdes{\begin{description}}
\def\enddes{\end{description}}
\def\begenu{\begin{enumerate}}
\def\begite{\begin{itemize}}
\def\endite{\end{itemize}}
\def\endenu{\end{enumerate}}

\def\lef[{\left[\begin{array}}
\def\rig]{\end{array}\right]}

\def\begcen{\begin{center}}
\def\endcen{\end{center}}
\def\begrem{\begin{remark}\rm}
\def\endrem{\end{remark}}
\def\begassum{\begin{assumption}}
\def\endassum{\end{assumption}}
\def\begassums{\begin{assumption*}}
\def\endassums{\end{assumption*}}
\def\begassu{\begin{ass}}
\def\endassu{\end{ass}}
\def\beglem{\begin{lemma}}
\def\endlem{\end{lemma}}
\def\begcor{\begin{corollary}}
\def\endcor{\end{corollary}}
\def\begfac{\begin{fact}}
\def\endfac{\end{fact}}

\def\liminft{\lim_{t \to \infty}}

\def\L2e{{\cal L}_{2e}}

\def\begsubequ{\begin{subequations}}
	\def\endsubequ{\end{subequations}}
\def\begpro{\begin{proposition}}
	\def\endpro{\end{proposition}}
\def\beglem{\begin{lemma}}
	\def\endlem{\end{lemma}}
\def\begass{\begin{assumption}}
	\def\endass{\end{assumption}}
\def\begcor{\begin{corollary}}
	\def\endcor{\end{corollary}}
\def\begproo{\begin{proof}}
	\def\endproo{\end{proof}}
\usepackage{wrapfig}

\usepackage{color}

\begin{document}

	\begin{frontmatter}
		\runtitle{Insert a suggested running title}
		\title{Globally Stable Discrete Time PID Passivity-based Control of Power Converters: Simulation and Experimental Results
		}
		
		\thanks[footnoteinfo]{Corresponding author: Wei He.}
		\author[London]{Alessio Moreschini}\ead{a.moreschini@imperial.ac.uk},
		\author[China]{Wei He}\ead{hwei@nuist.edu.cn},
		\author[MXC]{Romeo Ortega}\ead{romeo.ortega@itam.mx},
        \author[China]{Yiheng Lu}\ead{luyiheng@nuist.edu.cn},
        \author[China]{Tao Li}\ead{litaojia@nuist.edu.cn}
		
		\address[London]{Department of Electrical and Electronic Engineering, Imperial College London, SW7 2AZ London, United Kingdom}
		\address[China]{School of Automation, Nanjing University of Information Science and Technology, 210044  Nanjing, China}
		\address[MXC]{Department of Electrical and Electronic Engineering, ITAM, 01080 Ciudad de M\'{e}xico, Mexico}
		
		
		\begin{abstract}
The key idea behind PID Passivity-based Control (PID-PBC) is to leverage the passivity property of PIDs (for all positive gains) and wrap the PID controller around a passive output to ensure global stability in closed-loop. However, the practical applicability of PID-PBC is stymied by two key facts: (i) the vast majority of practical implementations of PIDs is carried-out in discrete time---discretizing the continuous time dynamical system of the PID; (ii) the well-known problem that passivity is not preserved upon discretization, even with small sampling times. Therefore, two aspects of the PID-PBC must be revisited for its safe practical application. First, we propose a discretization of the PID that ensures its passivity. Second, since the output that is identified as passive for the continuous time system is not necessarily passive for its discrete time version, we construct a new output that ensures the passivity property for the discretization of the system. In this paper, we provide a constructive answer to both issues for the case of power converter models. Instrumental to achieve this objective is the use of the implicit midpoint discretization method---which is a symplectic integration technique that preserves system invariants. Since the reference value for the output to be regulated in power converters is non-zero, we are henceforth interested in the property of passivity of the incremental model---currently known as shifted passivity. Therefore, we demonstrate that the resulting discrete-time PID-PBC defines a passive map for the incremental model and establish shifted passivity for the discretized power converter model. Combining these properties, we prove global stability for the feedback interconnection of the power converter with the discretized PID-PBC. The paper also presents simulations and experiments that demonstrate the performance of the proposed discretization.
		\end{abstract}
		
	\end{frontmatter}
	
\section{Introduction}
\label{sec1}
%
Switching power converters are, nowadays, an essential component of most electrical engineering applications. The ever increasing performance requirements on these devices translates into more stringent specifications on the quality of the converters control. The dynamics of power converters is highly nonlinear, even with fast switching, when their averaged model adequately describes their behavior~\cite{ESCVANORT,ORTetalbook}, and the validity of their linear approximation is restricted to a small neighborhood of the corresponding operating point~\cite{KASSCHVERbook}. Three additional difficulties pertaining to the problem of power converter control are  the following.	\begin{enumerate}
	\item[\textbf{(D1)}] The full state of the converter is usually {\em unknown}, because the implementation of sensors to measure the state is costly and noise-sensitive.
	\item[\textbf{(D2)}] The {\em parameters} of the converter are uncertain and/or time-varying. For instance,  the load of the power converter is highly uncertain, in general.
	\item[\textbf{(D3)}] The overwhelming majority of the practical implementations of power converter controllers is done in {\em discrete time} (DT), while their theoretical developments are carried-out in continuous time (CT)---based on the averaged model.  Hence, additional analysis is needed to assess the stability of the sampled closed-loop system.
\end{enumerate}

{These issues have been discussed extensively by several authors, including \cite{KAZetalbook,MONNOR1,MONNOR2,WANetalbook}} The vast majority of power converters in practical applications are controlled with classical PI loops---partially overcoming some of the issues above. Namely, the PI controller is usually wrapped around a current error signal, whose measurement poses no practical challenge.\footnote{Although very often it is a voltage signal that we want to regulate, people adopt this so-called {\em indirect method} of controlling the voltage through the current, to avoid the problem of unstable zero dynamics of the voltage output~\cite{ORTetalbook,SIRSILbook}.}  There is a widely accepted belief that, \emph{if} the PI is properly tuned, their behavior is acceptable even in the face of (slowly) time-varying changes of the converter parameters, including the converter load~\cite{KAZetalbook,WANetalbook}. The qualifier ``if" in the previous sentence is of outmost importance because, if the range of operation of the system to be controlled is ``wide"---as it is required in modern applications---the task of tuning the gains of a PI (or, for that matter, of any other controller for nonlinear systems) is far from obvious. Unfortunately, to date, the only systematic procedure to carry-out this task, is invoking standard {\em linear systems theory} arguments, \emph{e.g.}, pole placement, stability margins, and applying them to the linearized model of the converter. Various procedures to re-tune the PI gains, including gain-scheduling~\cite{VESILK}, relay auto-tuning~\cite{ASTHAGbook} and adaptation~\cite{ORTKEL}, have been proposed but they all suffer from well-documented serious limitations and drawbacks~\cite{ASTMURbook,RUGSHA}.

The PI gain tuning problem mentioned above has been (partially) overcome with the introduction in~\cite{HERetal} of the PID-PBC, whose main features are the following:
\begin{itemize}
\item The PID-PBC is applicable to a {\em large class} of power converters described by average bilinear models of the form~\eqref{eqn:power_converter}, see also~\cite[Appendix D.2.3]{ORTetalbook}.\footnote{As shown in~\cite{ORTetalbookpid} PID-PBC are applicable to a large class of physical systems including mechanical and electromechanical. For the sake of brevity in this article we restrict ourselves to power converter systems.}
\item \emph{Global stability} of the closed-loop is guaranteed \emph{for all} positive values of the PID tuning gains.
\item Although in its original formulations the PID-PBC assume the availability of the {\em full state} of the system, a recent observer-based versions of the scheme that only require the measurements used in standard PIDs have been reported in~\cite{BOBetal,HEetal}.
\item It is possible to incorporate to the controller an {\it adaptation} feature to estimate on-line some uncertain parameters---in particular the external load as done in ~\cite{BOBetal,HEetal}.
\item {The use of PID-PBC for the design of {\em inverters} has a very {\em close connection} with the widely popular {\em PQ Instantaneous Power} controllers for AC systems of~\cite{AKAbook}, which explains the wide acceptance that this new controller has had in the power electronics community---see ~\cite[Sec. 4.5.4]{ORTetalbook} and~\cite{ZONetal} for a thorough discussion on this issue.}
\end{itemize}

The approach adopted in the design of PID-PBC is in the line of~\cite{ORTetalcsm} and~\cite{ORTetalbook}, which relies on the use of energy balance concepts to control a system. Unlike most classical nonlinear control techniques found in the literature, which try to impose some predetermined dynamic behavior---usually through nonlinearity cancellation, domination of nonlinearites and high gain---energy-based methods exploit and respect the physical structure of the system. PBC is the generic name of this controller design methodology, which achieves stabilization by exploiting the passivity properties of the system. Due to the physical appealing that this methodology has, a vast literature has been advocated to its application to mechanical, electrical and electromechanical systems, see~\cite{ORTetalbook,VANbook}.

{\em Passivity} is the key property of power converters that is exploited in PID-PBC.  The first time that passivity principles were applied for power converter control is  in the foundational paper~\cite{SANVER}. It is well-known~\cite[Lemma 2.1]{ORTetalbookpid} that PID controllers define passive operators\footnote{It can be demonstrated that they satisfy the stronger output strict passivity property.}. Therefore, as a corollary of the Passivity Theorem~\cite{ORTetalbook,VANbook} we have that wrapping a PID around a passive output ensures input-output stability of the system and convergence to zero of the passive output---see~\cite{ORTetalbookpid} for a recent survey on PID-PBC theory and applications.\footnote{It should be underscored that the current use of PIDs in applications, wrapping the PID around a current error signal, does not ensure the aforementioned property, because this signal {\em is not} a passive output. See~\cite{MONetal} for the case study of a Voltage Sourced Inverter.} Passivity has provided a theoretical framework for controller design, which  is a topic of paramount importance. On one hand, it allows practitioners to apply the control law with confidence and, on the other hand, considerably simplifies the commissioning stage, which now that stability is guaranteed for all positive tuning gains, can concentrate on the {\em transient performance} specifications.

In many applications, including power converters, the control objective is to drive a given output to a value \emph{different} from zero, which is associated with the steady-state behavior of the system. In this case, we are interested in proving that the {\em incremental model} is passive---whence, driving the output of the incremental model to zero will ensure the signal of interest converges to its desired value. In~\cite{PERORTESP} it was shown that, for  a general class of models of power converters, it is possible to define an output signal such that the incremental model is passive---this result was later generalized in several directions in~\cite{HERetal}. This fundamental property, called ``passivity of the nonlinear incremental model" in~\cite{JAYetal}, is now  referred as \emph{shifted passivity}~\cite{VANbook} and has played a central role in many recent developments of the control community---see~\cite{KAWMORCUC,MONetalscl} and~\cite[Sec. 4.3]{ORTetalbookpid} for a recent characterization of port-Hamiltonian systems, which are shifted passive and~\cite{SIM} for the generalization to the dissipativity property. In the present context the main interest of this property is that, driving the shifted passive output to zero with a PID control ensures global stability of the desired equilibrium. Moreover, under a precise condition on the systems dissipation, that the  state of the converter system converges to its desired value.

{The PID controller of~\cite{HERetal} enjoys a very wide popularity.} However, it still suffers from the drawback {\bf (D3)} mentioned above. That is, the overwhelming majority of its practical implementations are done in {\em DT}---somehow discretizing the {\em continuous} dynamical system describing the PID. It is well-known that passivity is {\em generally not preserved} under the sampling process, even for {\em arbitrarily small} sampling times. Hence, additional analysis is needed to assess the stability of the sampled closed-loop system.  The phenomenon of loss of passivity under sampling has been extensively analyzed in various works, see {\em e.g.}~\cite{MORetal} for recent advancements and a literature survey. In these studies, the problem has been circumvented by “redefining the output” with the objective of guaranteeing a suitably defined passivity property for the newly-defined output, see, \emph{e.g.},~\cite{LAIAST,KAWMORCUC,KOTTHO,MAC,MATMORMONCYR,MONNOR3,MONNOR4,MORMATMONCYR,MORMONCYR,MORetal}.

This is also the approach adopted in this paper. In our case, we need to propose a discretization procedure for: (i) the CT PID and (ii) the CT model of the power converter. This discretization should be such that we are able to define outputs which are {\em shifted passive}---equivalently, find outputs such that the input-output maps of the incremental models are passive. Providing a solution to this challenging and, to the best of our knowledge, until now open problem is the {\em main contribution} of this paper.

The contributions of our work may be summarized as follows:
\begin{enumerate}
\item[{\bf C1}] We propose, for the first time, a discretization procedure for the design of the DT PID-PBC: the {\em implicit midpoint} rule, which is a well-known symplectic integration method that preserves the systems invariants of motion upon numerical integration~\cite{HAILUBWANbook}.
\item[{\bf C2}] It is shown that the DT PID obtained from the aforementioned discretization procedure defines a {\em passive map} for the incremental model---a property that is established with the incremental variation of the systems storage function.
\item[{\bf C3}] For the discretized power converter model we define an output with respect to which shifted passivity of the system is established, again using the incremental variation of the systems storage function.
\item[{\bf C4}] Combining these two latter properties, which exactly mimic the ones we have for the CT PID-PBC and the CT power converter model, we prove the main global stability result of the feedback interconnection of the power converter and the DT PID-PBC.
\end{enumerate}

The remainder of the paper is organized as follows. In Section~\ref{sec2} we present some preliminaries on power converters models and the CT PID-PBC of~\cite{HERetal} and briefly recall the midpoint discretization method. The passivity properties of the discretized power converter model and the DT PID-PBC are given in Section~\ref{sec3} and~\ref{sec4}, respectively. The main closed-loop stabilization result is presented in Section~\ref{sec5}.  Simulation and {\em experimental} results, which illustrate the {performance} of the proposed DT PI-PBC, are presented in Section~\ref{sec6}. Particular attention is given to two aspects: (i) the effect of increasing the sampling time and (ii) comparison with Euler discretization of the CT PID-PBC. The paper is wrapped-up with concluding remarks and future research in  Section~\ref{sec7}.\\

\noindent {\bf Notation.} $I_n$ is the $n \times n$ identity matrix.  For $x \in \R^n$,  we denote the Euclidean norm as $|x|^2:=x^\top x$. For $q \in \N$ we define the set $\bar q:=\{1,2,\dots,q\}$.  Throughout the paper we consider piece-wise constant inputs $u(t) = u_k$, $t_{k} \leq  t < t_{k+1}$, and a constant sampling time $\delta$, {\em i.e.}, $t_k = k\delta$, for $k\in\N$, an operation implemented via a standard zero-order hold. Given a sequence $(\cdot)_k \in \R^s$, a distinguished constant value $(\cdot)^\star  \in \R^s$ and a map $g:\R^s \to \R^m$ we denote $\tilde {(\cdot)}_k:=(\cdot)_k-(\cdot)^\star $ the error signal, $g^\star :=g((\cdot)^\star )$ the constant vector and the {\em forward difference operator} $\Delta$ as
$
\Delta g((\cdot)_k) \coloneqq g((\cdot)_{k+1}) - g((\cdot)_k).
$
%
\section{Preliminaries}
\label{sec2}
%
In this section we present some preliminaries on power converters models and the CT PID-PBC of~\cite{HERetal}, for further details see~\cite{ORTetalbookpid}. We also briefly recall the midpoint discretization method further elaborated in ~\cite{HAILUBWANbook,KAWMORCUC}.
\subsection{CT power converter models}
\label{subsec21}
As shown in~\cite{HERetal,ORTetalbook,SIRSILbook} the average dynamics of a large class of power converters\footnote{{Similarly to the vast majority of controllers proposed for power converters, throughout the paper it is assumed that the converter operates in continuous conduction mode.}}---operating with sufficiently fast sampling rate---is described by the port Hamiltonian model~\cite{VANbook}
	\begin{equation}
		\label{eqn:power_converter}
		\dot{x} = \Big(J_0 - R  + \sum_{i=1}^{m}u_iJ_i\Big) \frac{\displaystyle \partial^\top  H }{\displaystyle \partial x}(x) + \Big(G_0 + \sum_{i=1}^{m} u_iG_i \Big)E,
	\end{equation}
where $x(t)\in\R^n$ is the converter state vector, consisting of fluxes in inductors and charges in capacitors, $u(t)\in\R^m$ denotes the duty ratio of the switches. The total energy stored in inductors and capacitors is given by
\begin{equation}
\label{h}
H(x)=\frac{1}{2}x^\top Q x,
\end{equation}
with $Q = Q^\top > 0$ determined by the values of the capacitances and inductances, which are assumed positive. The matrices  $J_i = - J_i^\top$, $i\in \bar m$, are the interconnection matrices, $R = R^\top \geq 0$ represents the dissipation matrix due to the presence of resistors in the circuit, $G_0$ and $G_i,\;i \in \bar m,$ are $n\times n$ input matrices, and $E\in\R^n$ contains the external voltage and current sources, which may be switching---a feature that is captured by the matrices $G_i,\;i \in \bar m$. We make the important observation that all the matrices $Q, J_i, R, G_0$ and $G_i$ are {\em constant}.

The {\em power balance} equation for the system~\eqref{eqn:power_converter}---that quantifies the rate of change of the energy in time---is given by
$$
\underbrace{\dot H}_{\mbox{stored}}=- \underbrace{x^\top Q R Q x}_{\mbox{dissipated}}+\underbrace{x^\top Q \Big(G_0 + \sum_{i=1}^{m} u_iG_i \Big)E}_{\mbox{supplied\; power}}.
$$
Let us write the system the standard compact form as
\begin{equation}
	\label{eqn:power_converter_compact}
	\dot{x} = f(x)  + g(x)u,
\end{equation}
with
\begin{equation}
		\begin{aligned}
			f(x) &:= (J_0 - R)Qx  + G_0E, \\
			g(x) &:= \begin{bmatrix}J_1Qx + G_1 E  & | & \hdots & | & J_mQx + G_m E\end{bmatrix}.
			\label{fg}
		\end{aligned}
	\end{equation}
Recall that an {\em assignable} equilibrium $x^\star  \in \R^n$ for the system~\eqref{eqn:power_converter_compact} satisfies the algebraic equation
\begin{equation}\label{eq:admissible_equilibrium}
	0 = f^\star  + g^\star u^\star
\end{equation}
where $u^\star  \in \R^m$ is the corresponding equilibrium control.\footnote{\cite[Proposition B.1]{ORTetalbookpid} The set of assignable equilibria is defined as $\{x^\star \in\R^n|(g^\perp)^\star f^\star=0\}$, where $g^\perp : \R^n \to \R^{(n-m)\times n}$ is a full-rank left annihilator  of $g$. The (unique) corresponding equilibrium control is defined as $u^\star  = -[(g^\star)^\top g^\star ]^{-1}(g^\star)^\top f^\star$. }
\subsection{Definition of an output $y$ such that $\tilde u \mapsto \tilde y$ is passive}
\label{subsec22}
%
	To reveal some useful passivity properties of the system~\eqref{eqn:power_converter} it is shown in \cite[Equation (4.44)]{ORTetalbookpid} that, for every {\em assignable} equilibrium $x^\star  \in \R^n$, there exists an output map $y$ and an associated equilibrium output $y^\star$ such that the map $\tilde u \mapsto \tilde y $ is passive with an the {\em incremental} storage function $H(\tilde x)={1 \over 2} \tilde x^\top Q \tilde x$. Specifically, the system can be equivalently written in shifted form as
	\begin{align*}
		\dot{x} = (J_0 - R)Q \tilde x + g(x) \tilde u + (g(x) - g^*)u^*.
	\end{align*}
	Consequently, the derivative of the incremental storage function $H(\tilde x)={1 \over 2} \tilde x^\top Q \tilde x$ yields
	\begin{equation}
		\begin{aligned}
			\dot{H}(\tilde x) =&-\tilde x^\top QRQ \tilde x+ \tilde x^\top Q g(x) \tilde u.
		\end{aligned}
	\end{equation}
	We make now the observation that, exploiting the skew-symmetry ``properties" of $g(x)$, we can write
	$
	\tilde x^\top Q g(x)=\tilde x^\top Q g^\star .
	$
	Hence, by defining the incremental output as
	\begin{equation}
		\label{ynew}
	\tilde y = (g^\star)^\top Q \tilde x,
	\end{equation}
	with the associated equilibrium output
	$y^\star  = (g^\star)^\top Qx^\star $,
	we obtain that
	$$
	\dot{H}(\tilde x)=-\tilde x^\top QRQ \tilde x+ \tilde y^\top \tilde u \leq  \tilde y^\top \tilde u ,
	$$
	which establishes passivity of the map $\tilde u \mapsto \tilde y$. Therefore, the system is \emph{shifted passive}.
\subsection{Global stabilization of the PID-PBC}
\label{subsec23}
%
We recall that the control objective is to stabilize with a suitable CT PID-PBC an (assignable) equilibrium $x^\star $ of the system, namely we want to verify that
\begequ
\lab{concon}
\liminft x(t)=x^\star,
\endequ
for all initial conditions and all PID-PBC gains. To achieve this end, we propose the CT PID-PBC~\cite[Eq. 4.1]{ORTetalbookpid} described by the equations
\begin{equation}\label{eqn:PID}
	\begin{aligned}
		\dot{\xi} &= \tilde y \\
		u &= -K_P \tilde y  - K_I\xi - K_D\dot{\tilde y},
	\end{aligned}
\end{equation}
where $K_P>0$, $K_I>0$, and $K_D \geq 0$ are the tuning gains.

\begin{proposition}\em
	\label{pro1}
	Consider the system~\eqref{eqn:power_converter_compact}, with the desired state equilibrium $x^\star$, in closed-loop with the CT PID-PBC~\eqref{eqn:PID}, where $\tilde y$ is given by~\eqref{ynew}. \\
	
\noindent {\bf (i)} 	The equilibrium $(x^\star,\xi^\star )$ of the closed-loop system is {\em globally stable} for any $K_P>0$, $K_I>0$ and $K_D \geq 0$.

\noindent {\bf (ii)} If there exists an $\alpha>0$ such that
	\begin{equation}
	\label{daminj}
	R +  g^\star  K_P (g^\star)^\top  > \alpha I_n.
	\end{equation}
we have that \eqref{concon}  holds  for all initial conditions.
\end{proposition}

\begin{proof}
The stability of the closed-loop system is established with the Lyapunov function
\begin{equation}
\begin{aligned}
V(\tilde x,\tilde \xi)& \coloneqq H(\tilde x) +\frac{1}{2}\tilde \xi^\top K_I \tilde \xi+ \frac{1}{2} \tilde x^\top Q g^\star(g^\star)^\top Q \tilde x,
\label{hc}
 \end{aligned}
\end{equation}
where $\xi^\star=-K_I^{-1} u^\star $, which follows from \eqref{eqn:PID}, and we note that
$$
 \tilde x^\top Q g^\star(g^\star)^\top Q \tilde x=\tilde y ^\top K_D \tilde y.
$$
Some simple calculations show that the derivative of the Lyapunov function yields
\begin{equation*}
	\begin{aligned}
		\dot V&= -\tilde x^\top QRQ \tilde x  -\tilde y ^\top K_P \tilde y,
	\end{aligned}
\end{equation*}
from which global stability follows.
	
The asymptotic claim follows expressing the right hand term above as
$$
\tilde x^\top QRQ \tilde x+\tilde y ^\top K_P \tilde y = \tilde x^\top Q [R + g^\star  K_P (g^\star)^\top]Q \tilde x,
$$
exploiting the damping injection assumption~\eqref{daminj}, which ensures $\dot V \leq -\alpha |Q \tilde x|^2$, and invoking classical detectability arguments \cite[Theorem A.2]{ORTetalbookpid}.
\end{proof}
\subsection{Midpoint discretization}
\label{subsec24}
{
The dynamics of Hamiltonian systems preserve volume in phase space—a consequence of their underlying symplectic structure~\cite{HAILUBWANbook,MORMONCYR,KAWMORCUC,KOTTHO}. This structure, central to Hamiltonian mechanics, is associated with the conservation of phase space volume (Liouville’s theorem) and, in many cases, energy, depending on the Hamiltonian. In long-term simulations and digital control design, preserving these geometric properties is crucial. A numerical method is symplectic if it preserves the symplectic form---{\em i.e.}., the geometric structure of phase space. To this end, symplectic integrators, {\em e.g.}, midpoint methods, are commonly employed to construct structure-preserving discrete-time approximations.

Consider an autonomous continuous-time system
\begin{align*}
	\dot{x}(t) &= F(x(t)),\\
	 y(t) &= Cx(t),
\end{align*}
where $x(t)\in\R^n$ is the state, $y(t)\in\R^m$ is the output and $C \in \R^{m \times n}$ is a constant matrix. A discrete-time approximation of the dynamics is typically obtained by applying a forward finite-difference method to approximate the time derivative, namely
\begin{equation}\label{appr_xy}
	\dot{x}(t) \approx \frac{1}{\delta} (x_{k+1} - x_{k})\  \implies \ \dot{y}(t) \approx \frac{1}{\delta} {C}(x_{k+1} - x_{k}),
\end{equation}
where for all $k\in\N$, $\delta\in\N$ denotes a constant {\em sampling} time, $t_k \coloneqq k\delta $ denotes the sampling instant, $x_{k}$ approximate the state at the sampling instant $x(t_k)$ (\emph{i.e.}, $x_{k} \approx x(t_k)$). While the finite-difference approximation is straightforward, it does not generally preserve the underlying symplectic structure of Hamiltonian systems. To address this, the implicit midpoint rule, a second-order symplectic integrator, is often employed. This method updates the system state using the DT equation
\begin{equation} \label{eq:imr}
	x_{k+1} = x_k + \delta F\left( \frac{x_k + x_{k+1}}{2} \right),
\end{equation}
where the right-hand side evaluates the vector field $F$ at the midpoint between $x_k$ and $x_{k+1}$. Equation~\eqref{eq:imr} is implicit because $x_{k+1}$ appears on both sides, making the update a nonlinear algebraic equation that must be solved at each step. This can be done using iterative methods such as fixed-point iteration or Newton-Raphson, depending on the properties of $F$, see~\cite{HAILUBWANbook}. Despite this added complexity, the implicit midpoint rule offers significant benefits in preserving the structure and long-term behavior of Hamiltonian systems. In particular cases, such as the system studied in this work, we will later show that \eqref{eq:imr} can be solved {\em explicitly}, eliminating the need for iterative computation while still retaining the benefits of symplectic integration. To simplify notation and analysis, the midpoint argument in \eqref{eq:imr} is often denoted by
\begin{equation} \label{zk}
z_k={1 \over 2}(x_{k+1}+x_k),
\end{equation}
so that $z_k$ can be obtained as the solution of the implicit equation $z_k = x_k + \frac{\delta}{2}F(z_k)$. While introducing $z_k$ does not change the structure of the method, it is helpful in later analysis and can make certain expressions more compact.
}
%
\section{Shifted Passivity of the Discretized Power Converter Model}
\label{sec3}
%
In this section we give an answer to the first central question of the paper: derive a DT version of the power converter model~\eqref{eqn:power_converter}---equivalently~\eqref{eqn:power_converter_compact}---for which we can identify an {\em output} such that the associated {\em incremental model} is passive. Equivalently, find an output such the input-output map of the DT power converter is {\em shifted passive}. Not surprisingly, this task is achieved applying the midpoint discretization method to the standard PID structure and, in a natural way, looking at the incremental variation of the converters {\em energy function}~\eqref{h}.

\begin{proposition}
	\em
	\label{pro2}
	The DT model of the CT converter dynamics~\eqref{eqn:power_converter} obtained via midpoint discretization yields a shifted passive input-output map with respect to the output
	\begequ
	\label{y}
	y_k = (g^\star)^\top Q z_k,
	\endequ
	with $z_k$ as in \eqref{zk}, corresponding equilibrium output
	\begin{equation}
		\label{yk}
		y^\star=(g^\star)^\top Q x^\star,
	\end{equation}
	and storage function the {\em incremental function} $H(\tilde x_k)$.
\end{proposition}

\begin{proof}
Given that $z_k = \frac{1}{2}(x_{k+1} + x_k)$, the midpoint discretization of the CT converter dynamics~\eqref{eqn:power_converter} is given as
\begin{equation}\label{eq:discretized_model}
	\begin{aligned}
		x_{k+1} &= x_k + \delta f\left(z_k\right) + \delta g\left(z_k\right)u_k,
	\end{aligned}
\end{equation}
with $f(x)$ and $g(x)$ given in~\eqref{fg} with and $z_k$ found in \eqref{zk}. Note that the assignable equilibrium points $x^\star$ for the CT systems, that is, the vectors satisfying~\eqref{eq:admissible_equilibrium}, are also equilibria of~\eqref{eq:discretized_model}. Hence, for a fixed assignable equilibrium $x^\star $ for the system~\eqref{eq:discretized_model} we need to prove that the output $y$ given in~\eqref{y} is such that the map $\tilde u_k \mapsto \tilde y_k$ (with $\tilde y_k = y_k - y^\star$) is passive. We will establish the required passivity property with the {\em incremental energy function} given by $H(\tilde x_k)$. Towards this end, compute the variation of the incremental  energy function, that is given by applying the forward difference operator $\Delta$ to $H(\tilde x_k)$, \emph{i.e.}, for all $k\in\N$ we denote
$$
\Delta H(\tilde x_k) \coloneqq H(\tilde x_{k+1}) - H(\tilde x_k).
$$
After some simple manipulations we can prove that
\begin{equation}		\label{delh0}
	\begin{aligned}
		\Delta H(\tilde x_k) =& \tilde z_k^\top Q (\tilde x_{k+1}-\tilde x_k)\\
		=& \delta\tilde z_k^\top Q \left(f\left(z_k\right)  + g\left(z_k\right)u_k  \right)\\	
		=& \delta\tilde z_k^\top Q \left[f\left(z_k\right) - f^\star  + g\left(z_k\right) u_k - g^\star u^\star  \right]\\
		=& -\delta\tilde z_k^\top QRQ\tilde z_k  + \delta\tilde z_k^\top Q\left[g(z_k)u_k - g^\star u^\star  \right].
	\end{aligned}
\end{equation}
The last right hand term may be expressed as
\begin{align}
	&\tilde z_k^\top Q\left[g(z_k)u_k - g^\star u^\star  \right]  \nonumber\\
	= &\tilde z_k^\top Qg(z_k)u_k - \tilde z_k^\top Qg^\star u^\star   \nonumber\\
	=&  	\tilde z_k^\top Qg(z_k)u_k - \tilde z_k^\top Qg^\star u_k + \tilde z_k^\top Qg^\star u_k - \tilde z_k^\top Qg^\star u^\star \nonumber \\
	\label{last_eq12}
	=&  	\tilde z_k^\top Q \left[g(z_k) - g^\star \right]u_k + \tilde z_k^\top Qg^\star \tilde u_k.	
\end{align}
Note however that, from~\eqref{fg}, the following
\begin{align*}
	g(z_k) - g^\star  &= \sum_{i=1}^{m} \Big( J_iQz_k + G_i E \Big) - \Big( J_iQx^\star  + G_i E \Big) \\
	&=\sum_{i=1}^{m} J_iQ (z_k - x^\star )=\sum_{i=1}^{m} J_iQ \tilde z_k,
\end{align*}
holds. Therefore, using the skew-symmetric property of $J_i$, we have that $\tilde z_k^\top Q \left[g(z_k) - g^\star \right] = 0$. Hence,	
\begin{align}
	\frac{1}{\delta}\Delta H(\tilde x_k) & = -\tilde z_k^\top QRQ\tilde z_k+\tilde z_k^\top Qg^\star \tilde u_k \\
	& = -\tilde z_k^\top QRQ\tilde z_k+  \tilde y_k^\top \tilde u_k \leq \tilde y_k^\top \tilde u_k
	\label{delh}
\end{align}
where, to obtain the last identity, we have used \eqref{y}.
\end{proof}

We conclude this section by revisiting the observation from Section~\ref{subsec23} that, for the DT model of the power converter, it is possible to derive an {\em explicit} expression for the system dynamics without the need to solve an implicit function. Indeed, the model~\eqref{eq:discretized_model} can be equivalently expressed as
\begin{align}
\lab{exep}
x_{k+1} =\mathcal{A}(u_k) x_k + \mathcal{B}(u_k) E,
\end{align}
where the matrices $\mathcal{A}(u_k)$ and $\mathcal{B}(u_k)$ are defined by
\begin{align*}
	\mathcal{A}(u_k) &\coloneqq\left[I_n-\frac{\delta }{2}\mathcal{N}(u_k)Q\right]^{-1} \left[I_n+\frac{\delta }{2}\mathcal{N}(u_k)Q\right],\\
	\mathcal{B}(u_k) &\coloneqq \delta \left[I_n-\frac{\delta }{2}\mathcal{N}(u_k)Q\right]^{-1} \mathcal{M}(u_k),
\end{align*}
with $\mathcal{N}(u_k)$ and $\mathcal{M}(u_k)$ defined by
\begin{equation}\label{matricesNM}
	\begin{aligned}
		\mathcal{N}(u_k) &\coloneqq J_0 - R + \sum_{i=1}^{m}u_i(k)J_i, \\
		\mathcal{M}(u_k) &\coloneqq  G_0 +  \sum_{i=1}^{m} u_i(k)G_i.
	\end{aligned}
\end{equation}

%
\section{Shifted Passivity of the DT PID}
\label{sec4}
%
As explained in the introduction, in typical practical applications the power converter is realized physically and we apply for its control a numerically implemented DT version of the controller---in our case a PID. A key step in the design of CT PID-PBC is the identification of an output  such that its incremental model is passive, that is, a shifted passive output.\footnote{The qualifier ``shifted" is compulsory because, even though the desired value for the output $y^\star $ is equal to zero, the equilibrium input $u^\star  \neq 0$.} The motivation to identify this shifted passive output---which is the {\em raison d'\^etre} of PID-PBC---is that it is well-known that CT PID controllers define {\em output strictly passive maps}~\cite[Lemma 2.1]{ORTetalbookpid}. Consequently, closing the loop around the shifted passive output will ensure that output regulation is ensured---and under some additional detectability conditions that the desired equilibrium point is globally stabilized.

Since it is well-known that passivity is usually not preserved under sampling it is necessary to identify such an output for a DT version of the power converter model~\eqref{eqn:power_converter_compact}---a task that was carried-out in the previous section. Instrumental to establish this result was the use of the symplectic structure-preserving midpoint discretization method briefly explained in Section~\ref{subsec23}.  To complete the design we need to develop a DT version of the PID, for which we can prove its shifted passivity, a task that is done in this section.

\begin{proposition} \em
\label{pro3}
Consider the CT model of the PID~\eqref{eqn:PID}, its midpoint discretization with input the output~\eqref{y} is given by\footnote{The derivative term in the PID is handled proposing  $\frac{1}{\delta}\mathcal{C} (x_{k+1} - x_{k})$ as DT approximation of $\dot y$ in consistence with \eqref{appr_xy}.}
\begali{
	\nonumber
	\xi_{k+1} &= \xi_k + \delta \tilde y_k,\\
	u_k &= -K_P \tilde y_k - \frac{1}{2}K_I(\xi_{k+1} + \xi_k) - \frac{1}{\delta}K_D\mathcal{C} \Delta x_{k} ,
	\label{xiuk}	
}
where, to simplify the notation, we defined the constant matrix
\begin{equation}
\label{c}
\mathcal{C}:= (g^\star)^\top Q,
\end{equation}
and $K_P>0$, $K_I>0$, and $K_D \geq 0$ are the tuning gains. The map $\tilde y \mapsto - \tilde u $ associated to~\eqref{xiuk} is {\em output strictly passive} with storage function
\begin{equation}
\label{stofunpid}
H_c(\tilde \xi_k,\tilde x_k) = \frac{1}{2}\tilde \xi_k^\top K_I \tilde \xi_k + \frac{1}{2}\tilde x_k^\top \mathcal{C}^\top K_D\mathcal{C} \tilde x_k.
\end{equation}
\end{proposition}

\begin{proof}
To prove that the discretized PID is shifted passive we check the incremental variation of its storage function~\eqref{stofunpid}. That is,
{$$
\Delta H_c(\tilde \xi_k,\tilde x_k) \coloneqq H_c(\tilde \xi_{k+1},\tilde x_{k+1}) - H_c(\tilde \xi_{k}, \tilde x_{k}).
$$
} This yields that
\begin{equation*}
	\begin{aligned}
		\Delta H_c =& \frac{1}{2}(\tilde \xi_{k+1} -\tilde \xi_{k})^\top K_I (\tilde \xi_{k+1} +\tilde \xi_{k}) \\
		&+\frac{1}{2}(\tilde x_{k+1} +\tilde x_{k})^\top \mathcal{C}^\top K_D\mathcal{C} (\tilde x_{k+1} -\tilde x_{k})\\
		=& \frac{\delta}{2}\tilde y_k^\top   K_I (\tilde \xi_{k+1} +\tilde \xi_{k})+ y_k^\top K_D \mathcal{C} (\tilde x_{k+1} - \tilde x_k)\\
		=& \tilde y_k^\top \left( \frac{\delta}{2} K_I (\tilde \xi_{k+1} +\tilde \xi_{k})   + K_D \mathcal{C} (\tilde x_{k+1} -  \tilde x_k) \right)\\
		=& \tilde y_k^\top \left( \frac{\delta}{2} K_I (\tilde \xi_{k+1} +\tilde \xi_{k})   + K_D \mathcal{C} \Delta x_{k} \right),
	\end{aligned}
\end{equation*}
where $\tilde y_k = y_k - y^\star$. Now, computing the equilibria of \eqref{xiuk} we get
$$
u^\star =-K_I \xi^\star .
$$
Hence, we can write
\begalis{
\frac{\delta}{2}K_I(\tilde \xi_{k+1} + \tilde \xi_k) = \frac{\delta}{2}K_I(\xi_{k+1} + \xi_k) + \delta u^\star .
}
This, together with the fact that, from~\eqref{xiuk}, we have
\begin{align*}
\frac{\delta}{2}K_I(\tilde \xi_{k+1} + \tilde \xi_k) + K_D\mathcal{C} \Delta x_{k} 	 = -\delta \tilde u_k -\delta K_P \tilde y_k,
\end{align*}
which yields
\begin{equation}
	\label{delhc}
	\begin{aligned}
		\frac{1}{\delta} \Delta H_c =&  - \tilde y_k^\top  K_P \tilde y_k - \tilde y_k^\top \tilde u_k,
	\end{aligned}
\end{equation}
hence completing the proof.
\end{proof}

\section{Closed-Loop Stability}
\label{sec5}
%
In this section we present the main result of the paper pertaining to the stabilization of a desired equilibrium $x^\star $ of the power converter~\eqref{eqn:power_converter} with the DT PID-PBC \eqref{xiuk}. In particular, to control the power converter dynamics~\eqref{eqn:power_converter} with the DT PID-PBC~\eqref{xiuk}, we consider a digital-to-analog converter such that $u(t)$ is a piece-wise constant input given by $u(t) = u_k$, for all $t\in [k\delta, k\delta+\delta )$---an operation implemented via a standard zero-order hold.

{
\begin{proposition} \em
\label{pro4}
Consider the power converter dynamics~\eqref{eqn:power_converter} with desired equilibrium $x^\star $, and its DT model ~\eqref{eq:discretized_model} with output the shifted passive signal~\eqref{y}, in closed-loop with the DT PID-PBC~\eqref{xiuk}. \\

\noindent {\bf (i)} For all choices of the tuning gains $K_P>0$, $K_I>0$, and $K_D \geq 0$ we have that the equilibrium  of the closed loop system $(x^\star,\xi^\star )$ is {\em globally stable}.

\noindent {\bf (ii)} If the {\em damping injection assumption}\footnote{The damping injection condition for global asymptotic stability of the DT and the CT PID controllers coincide.} ~\eqref{daminj} is satisfied, we have that, $\liminf z_k=x^\star,$ for all initial conditions.
\end{proposition}

\begin{proof}
	Consider the Lyapunov function candidate
	$$
	V(\tilde x_k,\tilde \xi_k) \coloneqq \frac{1}{\delta}H(\tilde x_k)  + \frac{1}{\delta}H_c(\tilde \xi_k,\tilde x_k),
	$$
	which is clearly radially unbounded and positive definite with respect to the equilibrium $(x_k,\xi_k)=(x^\star,\xi^\star )$. Its variation is computed adding up \eqref{delh} and \eqref{delhc} as follows
	\begali{
		\nonumber
		\Delta V(\tilde x_k,\tilde \xi_k) &
		= -\tilde z_k^\top QRQ\tilde z_k+  \tilde y_k^\top \tilde u_k - \tilde y_k^\top  K_P \tilde y_k - \tilde y_k^\top \tilde u_k\\
		& = - \tilde z_k^\top Q\left[ R +  g^\star K_P (g^\star)^\top\right]Q \tilde z_k \leq 0.
		\lab{delv}
	}
	From the last inequality above, invoking \cite[Theorem 1.2]{BOFCARSCH}, we conclude global stability of the equilibrium $(x^\star,\xi^\star )$.
	
	Moreover, if the damping injection assumption \eqref{daminj} holds, we have that $\Delta V(\tilde x_k,\tilde \xi_k) < 0$ for all $\tilde z_k \neq 0$, hence we deduce that,	
	 $\liminf \tilde z_k=0$, for any pair $(x_0,\xi_0)$, which implies $\liminf \frac{1}{2}(x_{k+1} + x_{k}) - x^\star = x_\infty - x^\star  = 0$, completing the proof of claim {\bf (ii)}.
\end{proof}

We further note that statement {\bf (ii)} of Proposition \eqref{pro4} implies that $\liminf y_k= y^\star$ as $\liminf z_k= x^\star$.

}
%
\section{Simulation and Experimental Results}
\label{sec6}
In this section, the Buck-Boost converter will be considered as an application example to assess the performance of the proposed controller \eqref{xiuk} via simulation and experimental studies. {Note that we only consider the case of its unidirectional power flow in our paper.}

\subsection{DT model of Buck-Boost converter and problem formulation}
The circuit topology of a Buck-Boost converter is shown in Fig. \ref{bbcircuit}.
\begin{figure}[t!]
	\centering
	\includegraphics[scale=1]{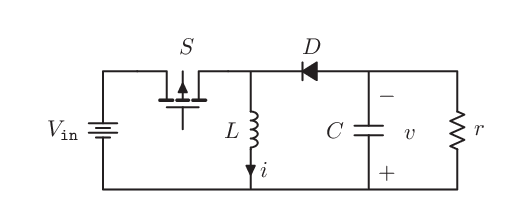}
	\caption{The circuit topology of buck-boost converter.}
	\label{bbcircuit}
\end{figure}
Its CT model can be expressed in the form \eqref{eqn:power_converter} with
\begin{align}
	&x=\left[
	\begin{array}{c}
		\phi \\
		q \\
	\end{array}
	\right], \
	J_0=\left[
	\begin{array}{cc}
		0 & -1 \\
		1 & 0 \\
	\end{array}
	\right], \
	J_1=\left[
	\begin{array}{cc}
		0 & 1 \\
		-1 & 0 \\
	\end{array}
	\right],\nonumber\\
	&R=\left[
	\begin{array}{cc}
		0 & 0 \\
		0 & 1/r \\
	\end{array}
	\right], \
	G_0=\left[
	\begin{array}{cc}
		0 & 0 \\
		0 & 0 \\
	\end{array}
	\right], \
	G_1=\left[
	\begin{array}{cc}
		1 & 0 \\
		0 & 0 \\
	\end{array}
	\right],\nonumber\\
	&E=\left[
	\begin{array}{c}
		V_{\tt in} \\
		0 \\
	\end{array}
	\right], \
	Q=\left[
	\begin{array}{cc}
		1/L & 0 \\
		0 & 1/C \\
	\end{array}
	\right],\nonumber
\end{align}
where $\phi$ is the flux in the inductor and $q$ is the charge in the capacitor, $V_{\tt in}$ is input voltage, $L$ and $C$ are the values of the inductor and
capacitor, and $r$ is the resistance load.

The relationship between the states and the signals indicated in Fig. \ref{bbcircuit} is given by
\begin{align}
	i=\frac{x_1}{L}, v=\frac{x_2}{C},\nonumber
\end{align}
where $i$ is inductor current and $v$ is output voltage. The converter dynamic model, in its compact form \eqref{eqn:power_converter_compact}, is expressed as
$$
\dot x = \underbrace{\begmat{-{1 \over C}x_2 \\ {1 \over L}x_1-{1 \over rC}x_2}}_{f(x)}+ \underbrace{\begmat{{V_{\tt in}}+{x_2 \over C}\\-{1 \over L}x_1 }}_{g(x)}u.
$$
Some simple calculations show that the assignable equilibrium set is given as
\begin{align}
	\lab{asseq}
	\mathcal{E}=\left\{x \in \mathbb{R}^2\ \Big| \ \frac{C^2V_{\tt in} x_1}{(x_2+V_{\tt in}C)}-\frac{Lx_2}{r}=0\right\}.
\end{align}
Consequently, for a given desired voltage reference $v^\star$, one obtains the reference
$$
x_1^\star=\frac{Lx_2^\star(x_2^\star+V_{\tt in}C)}{rC^2V_{\tt in}}.
$$
{The DT model of the Buck-Boost converter obtained via the midpoint discretisation method is given by
\begin{align}
x_{1, k+1}&=x_{1, k}+\delta\left(\frac{(u_{k}-1)z_{2, k}}{C}+V_{\tt in} u_k\right)\\
x_{2, k+1}&=x_{2, k}+\delta\left(\frac{(1-u_k)z_{1, k}}{L}-\frac{z_{2, k}}{rC}\right)
\end{align}
}

As stated in Section \ref{sec3}, the assignable equilibrium sets $\mathcal{E}$ of the CT and DT models coincide and is given by \eqref{asseq}.

Besides, we obtain that the passive output is defined as
$$
y_k=\Big(V_{\tt in}+\frac{x_2^\star}{C}\Big)\frac{z_{1, k}}{L}-\frac{x_1^\star z_{2, k}}{LC}
$$
and we get that $y^\star=\frac{V_{\tt in} x_1^\star}{L}$.
The control objective is to regulate the state $x_2$ around the equilibrium $x_2^\star$ using the proposed controller \eqref{xiuk}.

\begin{figure*}
	\centering
	\includegraphics[scale=0.7]{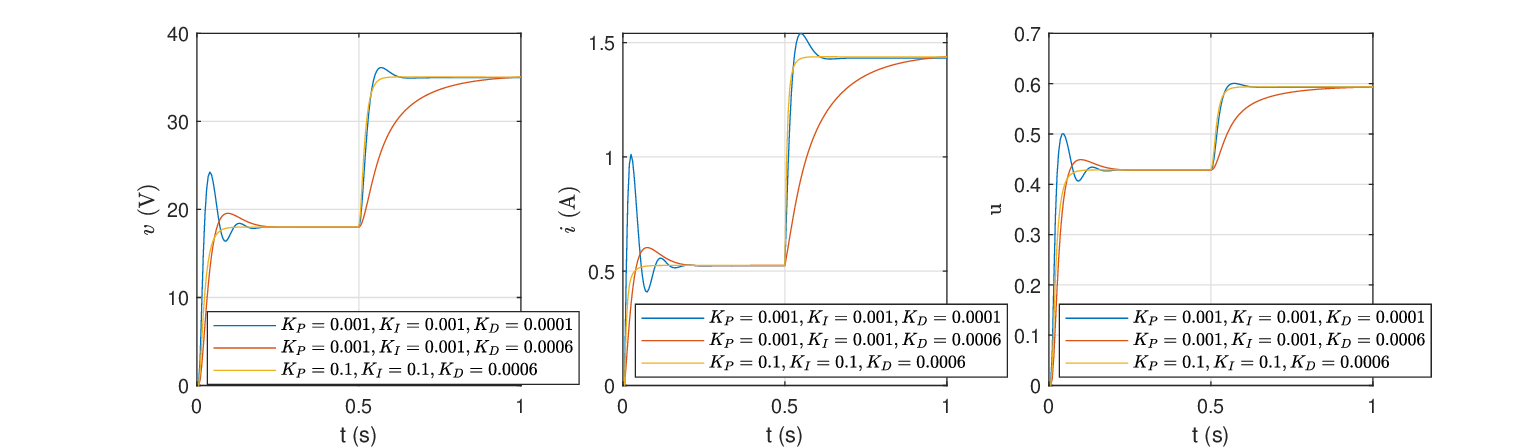}
	\caption{The response curves of the system with a step change in reference and different gains.}
	\label{fig1}
\end{figure*}
\begin{figure*}
	\centering
	\includegraphics[scale=0.7]{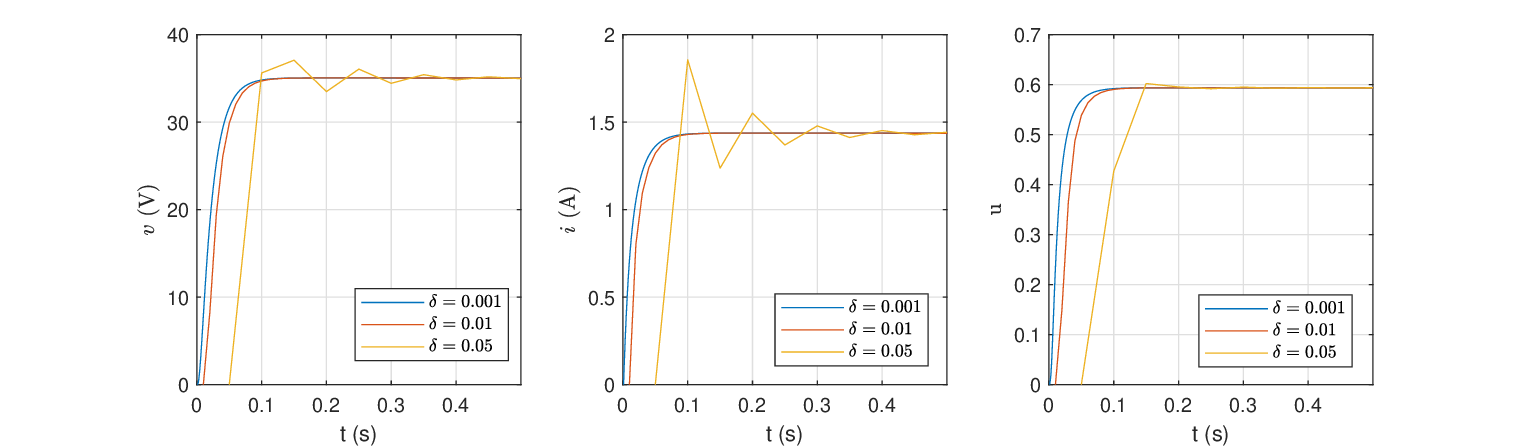}
	\caption{The response curves of the system with different sampling time $\delta$.}
	\label{fig2}
\end{figure*}
\begin{figure*}
	\centering
	\includegraphics[scale=0.7]{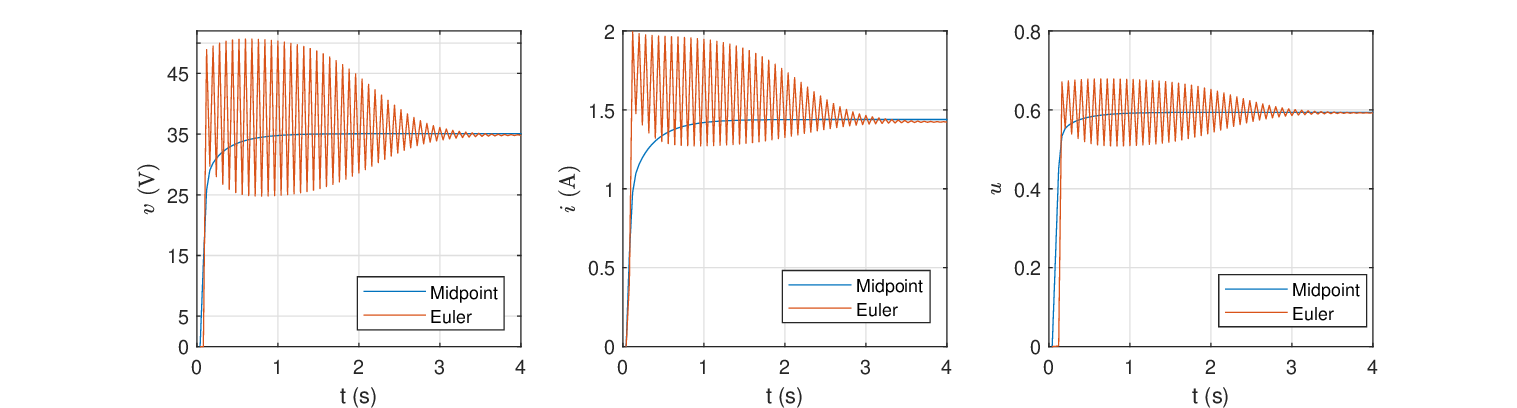}
	\caption{The response curves of the system under Euler discretization and the proposed midpoint method.}
	\label{fig3}
\end{figure*}

\subsection{Simulation results}
A simulation study of the Buck-Boost converter using its DT model is conducted to evaluate the performance of the DT PID-PBC in the ideal situation, where plant and controller are discretized.

The initial conditions are chosen as zero for all signals. The circuit parameters used in the simulation are provided in Table \ref{t1}.
\begin{table}[h!]
	\centering
	\begin{tabular}{lcc}
		\hline
		\hline
		Parameters & Symbols & Values\\
		\hline
		Input voltage & $V_{\tt in}$ & 24 V\\
		Inductance & $L$ & 1000 $\mu$H\\
		Capacitance & $C$ & 330 $\mu$F\\
        Resistance & $r$ & 60 $\Omega$\\
		\hline
		\hline
	\end{tabular}
	\caption{{Circuit parameters.}}
	\label{t1}
\end{table}

First, the tracking performance of the DT model in closed-loop with the proposed DT PID-PBC with different gains is validated. By fixing $\delta=5 \cdot 10^{-3}$, the response curves of the states and duty ratio are shown in Fig. \ref{fig1}, where the reference $v^\star$ is changed from $18$V to $35$V. It is observed that larger gains result in a smoother transient response and accelerate the convergence rate.

Second, we select different sampling times $\delta$ to show its effect on the closed-loop system. Here, $K_P=0.1, K_I=0.1, K_D=6\cdot 10^{-4}, v^\star=35$V are chosen.
The simulation results, presented in Fig. \ref{fig2}, show that a smaller sampling time $\delta$ leads to better transient performance. In contrast, a larger sampling time $\delta$ results in poorer transient and tracking performance, although stability is still maintained.

Finally, we compare the proposed control scheme with the controller discretized using the Euler method. Here, we used $K_P=10^{-4}, K_I=10^{-4}, K_D=10^{-3}, \delta=4\cdot 10^{-2}$ with voltage reference $v^\star= 35$V. As shown in Fig. \ref{fig3}, the transient performance of the Euler method is very poor. However, the trajectories converge to their desired values. It is important to note that choosing inappropriate control gains or setting $\delta$ too large---just slightly larger than $4\cdot 10^{-2}$---will lead to instability of the closed-loop system when using the Euler method. In contrast, the proposed symplectic-based digital controller ensures stability even for larger values of $\delta$.


\begin{figure}
	\centering
	\includegraphics[scale=0.6]{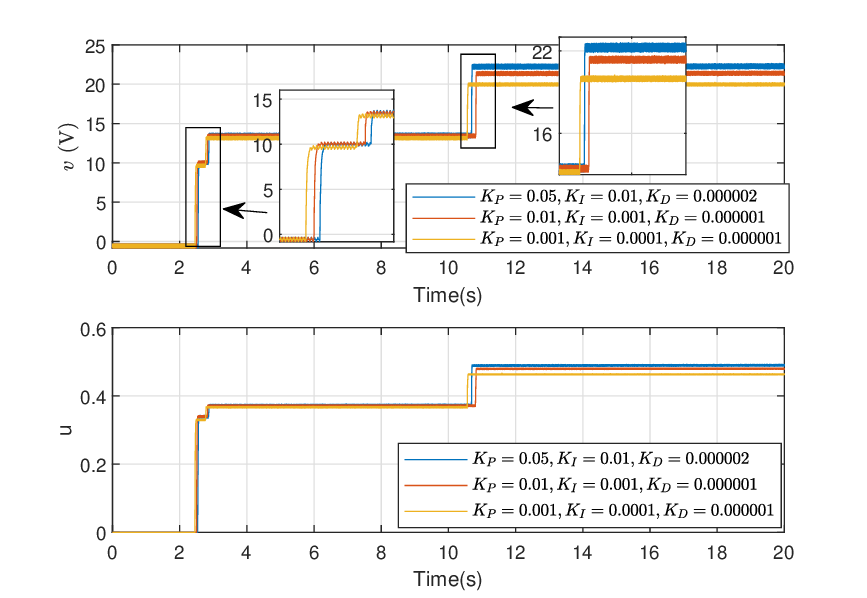}
	\caption{The experimental response curves of the system with step changes in reference from $15$ V to $22$ V.}
	\label{fig1-exp}
\end{figure}

\subsection{Experimental results}
An experimental study is conducted to further assess the control performance, using the same circuit parameters as those employed in the simulations and reported in Table~\ref{t1}. The control algorithm is first implemented in Matlab/Simulink, then compiled into a $C$ program and loaded onto the YXSPACE controller, which uses the TMS320F28335 central processing unit. The controller's input and output ports are used to regulate the output voltage of the converter to the desired value. Note that the Newton-Raphson method is used to solve implicit equations in the experimental study.

\begin{figure}
	\centering
	\includegraphics[scale=0.6]{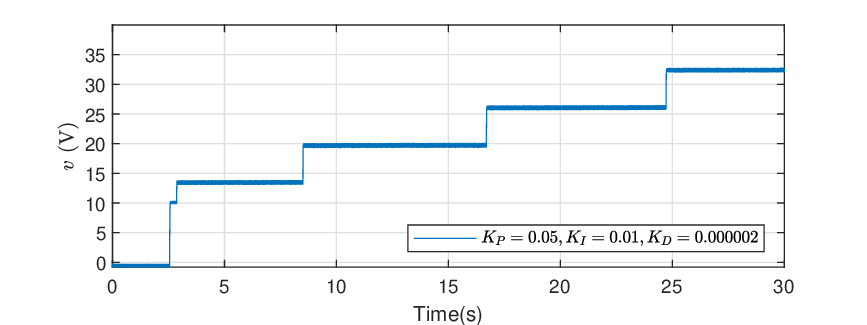}
	\caption{The experimental response curves of the system with step changes in reference from $15$ V to $20$ V to $25$ V to $30$ V.}
	\label{fig2-exp}
\end{figure}

We first evaluate the tracking performance by selecting different reference values, $v^\star$, ranging from $15V$ to $22V$. The experimental results are shown in Fig. \ref{fig1-exp}, with a sampling time $\delta = 5 \cdot 10^{-5}$ and varying controller gains. {The tentative selection of the gains is from the simulation study and their tuning objective is to obtain a nice performance. From the figure, it can be seen that the transient response is nearly identical for different gains. However, larger $K_P, K_I, K_D$ gains result in a small steady-state error, which is attributed to unmodeled dynamics.} Next, we test the system performance in both buck and boost modes by selecting additional reference values. The results, shown in Fig. \ref{fig2-exp}, demonstrate that the proposed method provides excellent tracking performance in both operational modes. We then assess the control performance under a step change in the reference voltage, transitioning from $15V$ to $22V$, for different sampling times. The controller gains used are $K_P = 10^{-3}$, $K_I = 10^{-5}$, and $K_D = 10^{-6}$. As shown in Fig. \ref{fig3-exp}, the system remains stable across all values of $\delta$, though the steady-state error becomes more evident as $\delta$ increases.

\begin{figure}
	\centering
	\includegraphics[scale=0.6]{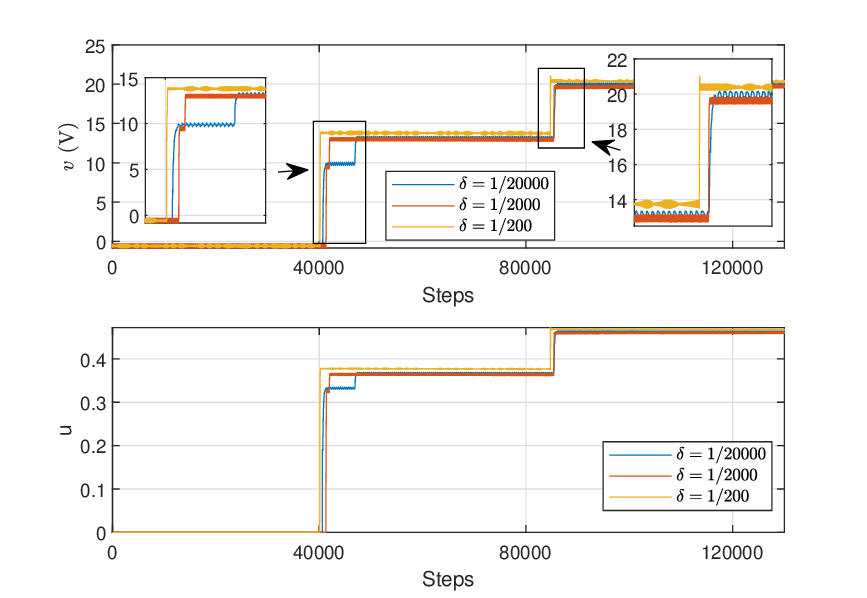}
	\caption{The experimental response curves of the system with different sampling time $\delta$.}
	\label{fig3-exp}
\end{figure}

%
%
\section{Concluding Remarks and Future Research}
\label{sec7}
%

We have addressed the problem of constructing digital PID Passivity-based Control (PID-PBC) to power converter models, addressing two critical issues that hinder its practical applicability. The first issue is that most practical implementations of PIDs are carried out in discrete time, which involves discretizing the continuous-time dynamical system of the PID. The second issue is the well-known problem that passivity is not preserved upon discretization, even with small sampling times.

To ensure the safe practical application of PID-PBC, we proposed a discretization of the PID that maintains its passivity and constructed outputs that are shifted passive, ensuring that the input-output maps of the incremental models remain passive. In particular, given that the reference value for the output to be regulated in power converters is non-zero, we focused on the property of passivity of the incremental model, known as shifted passivity. We demonstrated that the resulting discrete-time PID-PBC defines a passive map for the incremental model and established shifted passivity for the discretized power converter model. By combining these properties---which mirror those of the CT PID-PBC and the CT power converter model---we proved the main global stability result for the feedback interconnection of the power converter and the DT PID-PBC.

These results have been validated by means of simulations and experiments using a Buck-Boost converter with the control objective of driving a given output to a desired value, typically \emph{different} from zero. We have showed that the proposed method outperforms the typical emulation design, which is obtained simply performing the Euler discretization on both the {controller} and the plant.

The crux of our proposed method lies in redefining a passive output that preserves the shifted-passivity property after sampling, a common approach in passivity analysis under sampling. Future research will focus on developing the DT PID-PBC in conjunction with the recent $\varrho$-passivity theory~\cite{MORetal,MORetal2}, particularly by utilizing only the plant’s output measurements to construct a passive PID-PBC controller in a sampled-data fashion, without the need of discretizing the plant to be controlled.

{A second task that we plan to accomplish in the near future is to add a parameter estimator for some of the converter parameters, for instance, the load resistance $r$. It is clear that, the definition of the desired value of the states depends on these (usually) uncertain parameters. This task is carried out for the continuous-time PID-PBC in \cite{BOBetal}.}

\section*{Acknowledgments}
The second author would like to thank Mohammad Masoud Namazi's discussion about the experiment codes.

\end{document}